\documentclass[11pt,letterpaper]{article}
\usepackage{graphicx}
\usepackage[letterpaper,margin=1in]{geometry}
\usepackage{physics,complexity}
\usepackage{amsmath,amssymb,amsthm,mathrsfs}
\usepackage[colorlinks = true]{hyperref}
\usepackage{longtable}
\usepackage[dvipsnames]{xcolor}
\usepackage{caption,subcaption}
\usepackage[nameinlink, capitalize]{cleveref}
\newcommand{\mc}{\mathcal}
\usepackage[backend=biber,style=alphabetic,url=false,isbn=false,maxnames=10,minnames=3,maxalphanames=4,minalphanames=3]{biblatex}
\usepackage{tikz}
\usetikzlibrary{quantikz2}
\usetikzlibrary{decorations.pathreplacing}
\pgfmathsetmacro\MathAxis{height("$\vcenter{}$")}
\newcommand{\tikzlength}{0.7}

\usepackage{xcolor}
\definecolor{darkred}  {rgb}{0.5,0,0}
\definecolor{darkblue} {rgb}{0,0,0.5}
\definecolor{darkgreen}{rgb}{0,0.5,0}

\hypersetup{
  urlcolor   = blue,         
  linkcolor  = darkblue,     
  citecolor  = darkgreen,    
  filecolor   = darkred       
}

\newtheorem{dfn}{Definition}
\newtheorem{fact}{Fact}
\newtheorem{thm}{Theorem}
\newtheorem{cor}{Corollary}
\newtheorem{lem}{Lemma}
\newtheorem{remark}{Remark}
\title{Learning quantum states prepared by shallow circuits\\ in polynomial time}

\author{Zeph Landau\thanks{UC Berkeley. \href{mailto:zeph.landau@gmail.com}{zeph.landau@gmail.com}}
\and Yunchao Liu\thanks{UC Berkeley and Harvard University. \href{mailto:yunchaoliu@berkeley.edu}{yunchaoliu@berkeley.edu}}}

\begin{document}
\date{\vspace{-10mm}}
\maketitle
\begin{abstract}
    We give a polynomial time algorithm that, given copies of an unknown quantum state $\ket{\psi}=U\ket{0^n}$ that is prepared by an unknown constant depth circuit $U$ on a finite-dimensional lattice, learns a constant depth quantum circuit that prepares $\ket{\psi}$. The algorithm extends to the case when the depth of $U$ is $\mathrm{polylog}(n)$, with a quasi-polynomial run-time. The key new idea is a simple and general procedure that efficiently reconstructs the global state $\ket{\psi}$ from its local reduced density matrices. As an application, we give an efficient algorithm to test whether an unknown quantum state on a lattice has low or high quantum circuit complexity.
\end{abstract}

\section{Introduction}
We consider the following fundamental question in quantum complexity theory: is there an efficient algorithm to learn quantum states of low circuit complexity? That is, given copies of an unknown quantum state $\ket{\psi}$ that was generated from a constant depth quantum circuit acting on a finite-dimensional lattice, whether there exists an efficient algorithm that learns a circuit to prepare $\ket{\psi}$. In this paper we answer this question in the affirmative.

\begin{thm}[Main result, simplified version of~\cref{thm:maindetailed}]
\label{thm:main}
    There is an algorithm that, given copies of an unknown state $\ket{\psi}$, with the promise that $\ket{\psi}=U\ket{0^n}$ where $U$ is an unknown depth-$d$ unitary circuit acting on a $k$-dimensional lattice (using arbitrary 2-qubit gates), outputs a depth-$(2k+1)d$ unitary circuit $W$ that prepares $\ket{\psi}$ up to $0.01$ trace distance, with success probability $0.99$. The algorithm uses $M$ copies of $\ket{\psi}$ and runs in time $T$, where
    \begin{equation}
        M=\Tilde{O}(n^4)\cdot 2^{O(c)},\quad T=\Tilde{O}(n^4)\cdot 2^{O(c)}+\left(nkd\cdot c\right)^{O(d\cdot c)}.
    \end{equation}
Here, $c=O((3k)^{k+2} d)^k$, and $W$ uses $r\cdot n$ ancilla qubits where $r>0$ can be chosen to be an arbitrarily small constant.
\end{thm}
Note that the running time is polynomial when $d=O(1)$, and quasi-polynomial when $d=\mathrm{polylog}(n)$. In addition, the dominating term in the running time (second term in $T$) can be significantly improved when assuming a discrete gate set.

For quantum states prepared by shallow circuits, it is known that learning (sufficiently large) local reduced density matrices suffices to information-theoretically reconstruct the state~\cite{rouze2021learning,yu2023learning}. The question is whether the reconstruction can be computationally efficient. A naive approach finds small circuits for different local regions and stitch them together into a global circuit by checking local consistency, but this runs into a seemingly hard constraint satisfaction problem (in two and higher dimensions). Recent work~\cite{huang2024learning} developed an efficient algorithm in 2D by showing that to learn a 2D state it suffices to solve a 1D constraint satisfaction problem which is efficient; however this approach runs into the same issue at three and higher dimensions. Here we develop new techniques that do not rely on solving any consistency problem; this enables our algorithm to work at arbitrary dimensions.

\paragraph{Testing quantum circuit complexity.} As an application of our result, we also give an algorithm to test whether an unknown state on a $k$-dimensional lattice has low or high quantum circuit complexity.

\begin{cor}[Simplified version of~\cref{theorem:testing}]
\label{cor:testing}
    Fix some constant $L>0$. Given copies of an unknown state $\ket{\psi}$ on a $k$-dimensional lattice, with the promise that one of the following two cases hold:
    \begin{itemize}
    \item \textbf{Case 1: Low complexity.} $\ket{\psi}=U\ket{0^n}$ where the depth of $U$ is at most $L$;
    \item \textbf{Case 2: High complexity.} 
    Any state prepared by a constant depth circuit using $O(n)$ ancilla qubits is at least $0.01$-far from $\ket{\psi}$ in trace distance.
\end{itemize}
There is an algorithm that decides which is the case, with success probability at least 0.99, with polynomial sample and time complexity.
\end{cor}

Next we give a brief review of the background and motivations of this work.

\paragraph{Learning phases of matter.} Quantum systems defined on finite-dimensional lattices are a central subject in condensed matter physics, where quantum states are classified into different ``phases of matter''~\cite{Chen2010Local}. Quantum circuit complexity plays an important role in the definition of phases of matter: the ``trivial'' phase is typically defined as quantum states prepared by a constant (or $\mathrm{polylog}(n)$) depth circuit acting on a $k$-dimensional lattice (e.g.~\cite{aharonov2018quantum,Piroli2021Quantum}), while quantum states in a topologically ordered phase have high circuit complexity~\cite{Bravyi2006LiebRobinson}. Our result therefore shows that:
\begin{itemize}
    \item Quantum states in the ``trivial'' phase can be learned in polynomial (or quasi-polynomial) time.
    \item Given an arbitrary quantum state, there is an efficient algorithm to test whether it is in the ``trivial'' phase or some other high-complexity phase.
\end{itemize}

\paragraph{Quantum algorithms in NISQ.} Shallow quantum circuits provide a natural theoretical model for noisy intermediate-scale quantum (NISQ) computation, and optimism about quantum advantage from NISQ devices relies on the fact that shallow quantum circuits can generate nontrivial probability distributions that are hard to simulate classically~\cite{terhal2004adaptive,gao2017quantum,BermejoVega2018architecture,Hangleiter2023Computational}. This motivates many NISQ algorithms which heuristically use shallow quantum circuits as an ansatz; the goal is to optimize a parameterized shallow quantum circuit to solve interesting problems~\cite{Bharti2022Noisy}. In practice, these algorithms run into issues such as local minima~\cite{anschuetz2022quantum}. Our learning algorithm addresses a simple question in this direction: provably learning a shallow quantum circuit to prepare an unknown quantum state, which potentially provides a useful primitive for new NISQ algorithms.

\paragraph{Other related works.} The classical question of learning low-depth Boolean circuits (e.g. $\NC^0$ and $\AC^0$) is well-studied~\cite{mossel2003learning,linial1993constant,carmosino2016learning}. Recent work~\cite{huang2024learning} also gave an efficient algorithm for learning shallow quantum circuits from random input/output samples, which can be viewed as a quantum analog of learning shallow Boolean circuits. Our problem of learning quantum states prepared by shallow circuits can be viewed as another analog, but here there is no access of ``input'' to the shallow circuit and the learning algorithm can only make measurements to the quantum state.

Separately, developing efficient algorithms to learn interesting families of quantum states is a major direction in quantum learning theory, and we refer to~\cite{anshu2023survey} for a recent survey.

\paragraph{Discussion.} An interesting question is whether our algorithm can be generalized to other geometries (or interaction graphs) beyond finite-dimensional lattices. In fact, we show that our algorithm works for any geometry which has a property called ``covering scheme'' (\cref{def:coveringscheme}), and we construct covering schemes for finite-dimensional lattices. It is interesting to study what other geometries can have this covering scheme.



\paragraph{Note added.} We recently became aware of independent related work of Hyun-Soo Kim, Isaac Kim, and Daniel Ranard~\cite{privatecommunication}, achieving similar results via a different approach.




\section{Learning shallow quantum circuits vs quantum states}
To understand our problem, it is helpful to discuss the closely related problem of learning shallow quantum circuits: given query access to an unknown unitary $U$ that implements a shallow quantum circuit, learn a description of a shallow quantum circuit that is close to $U$. A previous work by us and coauthors gave an efficient algorithm for this problem~\cite{huang2024learning}. The nice thing about this algorithm is that it is very simple and yet works in general. Here's the proof: for any unitary $U$, the following identity holds.
\begin{equation}\label{eq:learningshallowcircuit}
    U\otimes U^\dag = \left(\prod_{i=1}^n S_i\right)\cdot\prod_{i=1}^n \left(U^\dag S_i U\right).
\end{equation}
Here, we have $n$ system qubits and $n$ ancilla qubits, and $S_i$ is the SWAP gate on the $i$-th system qubit and the $i$-th ancilla qubit. Each $U$, $U^\dag$ on RHS acts on the system. To see this identity, we first cancel $U$ with $U^\dag$ in the product, then RHS becomes $U^\dag \left(\prod_{i=1}^n S_i\right) U$. Then, \cref{eq:learningshallowcircuit} follows from this diagram.

\begin{equation}
    \begin{tikzpicture}[baseline={(0, 0.5*0.6*\tikzlength cm-\MathAxis pt)},x=0.6*\tikzlength cm,y=0.6*\tikzlength cm]
        \draw[black] (0.1,0) rectangle node{$U$} (2.9,1);
        \draw[black] (3.1,0) rectangle node{$U^\dag$} (5.9,1);
  \foreach \x in {1,2,...,6} {
        \draw[black] (\x-0.5, 0) -- (\x-0.5,-0.5);
        \draw[black] (\x-0.5, 1) -- (\x-0.5,1.5);
    }
    \end{tikzpicture}
\quad=\quad
\begin{tikzpicture}[baseline={(0, 0.5*0.6*\tikzlength cm-\MathAxis pt)},x=0.6*\tikzlength cm,y=0.6*\tikzlength cm]
        \draw[black] (0.1,0) rectangle node{$U^\dag$} (2.9,1);
        \draw[black] (0.1,-1) rectangle node{$U$} (2.9,-2);
  \foreach \x in {1,2,3} {
        \draw[black] (\x+2.5, 1) -- (\x+2.5,0);
        \draw[black] (\x+2.5, -1) -- (\x+2.5,-2.5);
        \draw[black] (\x-0.5, -2) -- (\x-0.5,-2.5);
    }

\foreach \x in {1,2,3} {
        \draw[black] (\x-0.5,1)..controls (\x-0.5,1.5) and (\x+2.5,1.5)..(\x+2.5,2);
        \draw[black] (\x-0.5,2)..controls (\x-0.5,1.5) and (\x+2.5,1.5)..(\x+2.5,1);
    }
    \foreach \x in {1,2,3} {
        \draw[black] (\x-0.5,-1)..controls (\x-0.5,-0.5) and (\x+2.5,-0.5)..(\x+2.5,0);
        \draw[black] (\x-0.5,0)..controls (\x-0.5,-0.5) and (\x+2.5,-0.5)..(\x+2.5,-1);
    }
    \end{tikzpicture}
\end{equation}
The learning algorithm directly follows from \cref{eq:learningshallowcircuit}: when $U$ is a shallow circuit, each operator $U^\dag S_i U$ is local and therefore easy to learn. Moreover, they also commute with each other, so we just learn each of them, multiply together arbitrarily, and we have learned a circuit that implements $U\otimes U^\dag$. If we arrange the ordering a bit smarter, we have a shallow circuit that implements $U\otimes U^\dag$. That is, the learned circuit is a shallow quantum circuit acting on $2n$ qubits, and it implements a unitary that is close to $U\otimes U^\dag$ in diamond distance.

So what about the current problem of learning quantum states prepared by shallow circuits? Clearly, what we want to achieve is to find a similarly simple technique that works in general. However, that turned out to be not so easy, because there are fundamental distinctions between the two problems. On a high level, the problem of learning shallow circuits has stronger access (we can query the unknown unitary with desired input and measure the output in a desired basis), and also has stronger requirement (we need to learn the circuit, instead of preparing a specific state), and seems incomparable with learning quantum states prepared by shallow circuits. 

But if one thinks harder, one starts to realize that the stronger access really makes the problem much easier. The structure demonstrated in \cref{eq:learningshallowcircuit} is really special: a shallow quantum circuit can be directly decomposed into a product of commuting local unitaries, which are also local observables that are easy to learn. This only exists because we have access to both the input and the output. But we are just given copies of a specific quantum state for the current problem, and the dream would be to find a simple and efficient procedure to reconstruct the state from local observables (or local reduced density matrices). However, we cannot hope to get something really similar to \cref{eq:learningshallowcircuit}: in general, there is no way to directly decompose a quantum state into reduced density matrices. Luckily, it turns out that there is a simple technique that works in general for our problem, which is presented in this work.


\paragraph{Our approach.} We start by introducing a new framework to reconstruct the state. Let $\ket{\psi}$ be an unknown quantum state prepared by an unknown shallow circuit. Our goal will be to learn \emph{local CPTP maps} $\mathcal{R}_1,\mathcal{R}_2,\dots,\mathcal{R}_L$ which fix the state, i.e., $\mathcal{R}_i(\ketbra{\psi}) = \ketbra{\psi}, \forall i\in[L]$. These maps are easy to obtain from local reduced density matrices. By definition, we have
\begin{equation}\label{eq:projection}
    \mathcal{R}_L\circ\mathcal{R}_{L-1}\cdots \circ\mathcal{R}_1(\ketbra{\psi}) = \ketbra{\psi}.
\end{equation}
Here's our plan. Note that in \cref{eq:projection}, in LHS we have an unknown input state $\ket{\psi}$, a sequence of known (learned) maps, and the output equals $\ket{\psi}$. We will show that when these local maps are chosen in a very careful way, the output in fact does not depend on the input. That is, $\mathcal{R}_L\circ\mathcal{R}_{L-1}\cdots \circ\mathcal{R}_1(\ketbra{\phi}) = \ketbra{\psi}$ for any state $\ket{\phi}$. This shows that we can already reconstruct $\ket{\psi}$ using these local maps. Moreover, we further show that we can \emph{directly read out} a shallow unitary quantum circuit $C$ from these maps that prepares the state within small trace distance, that is, 
\begin{equation}
    C\ket{0^n}\otimes \ket{0^m}\approx\ket{\psi}\otimes\ket{\mathrm{junk}},
\end{equation}
which uses $m$ ancilla qubits to prepare $\ket{\psi}$. Interestingly, unlike \cref{eq:learningshallowcircuit} which uses \emph{exactly} $n$ ancilla qubits, here $m$ is a tunable parameter which can be a small fraction of $n$.


\section{Learning algorithm}
\label{sec:learningalgorithm}

\paragraph{Notations.} Our goal is to learn a quantum state $\ket{\psi}$, with the promise that $\ket{\psi}=U\ket{0^n}$ where $U$ is an unknown depth-$d$ circuit acting on a $k$-dimensional lattice. We do not assume knowledge of the circuit architecture: each layer of the circuit consists of non-overlapping 2-qubit gates, where each qubit could interact with any of its neighbors. The following concepts are used throughout the argument.

\begin{itemize}
    \item Ball: $\mc B(A,r)$ denotes the radius-$r$ neighborhood on the lattice for a set of vertices $A$ (including $A$). For example, the dotted region in \cref{fig:2dinversion} shows a ball around region $A$ on the 2D lattice.
    \item Lightcone: $\mc L(A,d)$ denotes the volume of locations that can be reached by a depth $d$ circuit starting from region $A$. In particular, the support of the lightcone at top layer equals $\mc B(A,d)$. For example, the green region in~\cref{fig:1dinversion} denotes the lightcone of the leftmost qubit for a circuit defined on 1D lattice.
\end{itemize}
A lightcone can be viewed as propagating or spreading causal influence in a circuit from input to output. Later we will define a dual notion of \emph{backward} lightcone, which spreads from output to input.

\begin{figure}[t]
    \centering
\begin{subfigure}[b]{0.4\textwidth}
\centering
    \begin{equation*}
    \begin{tikzpicture}[baseline={(0, 0.5*\tikzlength cm-\MathAxis pt)},x=\tikzlength cm,y=\tikzlength cm]
    \fill[ForestGreen!30!white] (0,0) -- (1,0) -- (2,1) -- (0,1) -- (0,0);
        \draw[black] (0,0) rectangle node{$U$} (6,1);
  \foreach \x in {1,2,...,6} {
        \draw[black] (\x-0.5, 0) -- (\x-0.5,-0.4);
        \draw[black] (\x-0.5, 1) -- (\x-0.5,1.4);
        \node at (\x-0.5,-0.4)[anchor=north]{$\ket{0}$};
    }
    
    \draw[black] (0,1.4) -- (2,1.4) -- (1,2.4) -- (0,2.4) -- (0,1.4);
    \node at (0.7,1.9){$V_1$};
    \draw[black] (0.5,2.4) -- (0.5,2.8);
    \draw[black] (1.5,1.9) -- (1.5,2.8);
    \end{tikzpicture}\,\,=\,\, \ket{0}\otimes\ket{\phi}
\end{equation*}
    \caption{Local inversion in 1D}
    \label{fig:1dinversion}
    \end{subfigure}
\quad\quad
    \begin{subfigure}[b]{0.4\textwidth}
\centering
    \begin{tikzpicture}[baseline={(0, 4*\tikzlength cm-\MathAxis pt)},x=0.8*\tikzlength cm,y=0.8*\tikzlength cm]
    \fill[blue!30!white] (3,3) rectangle (5,5);
    \draw[black] (3,3) rectangle (5,5);
    \draw[black,densely dashed] (2.5,3) -- (3,2.5) -- (5,2.5) -- (5.5,3) -- (5.5,5) -- (5,5.5) -- (3,5.5) -- (2.5,5) -- (2.5,3);
    \draw[black] (0,0) rectangle (8,8);

\node at (2.75,4) {$B$};
\node at (4,4) {$A$};
\node at (4,6.5) {$C$};
\end{tikzpicture}
    \caption{Local inversion in 2D}
    \label{fig:2dinversion}
    \end{subfigure}

    \caption{Local inversions for quantum states prepared by shallow circuits. (a) 1D lattice; (b) 2D lattice, where a local inversion of $A$ can be constructed by applying a depth-$d$ circuit on $AB$.}
    \label{fig:localinversion}
\end{figure}
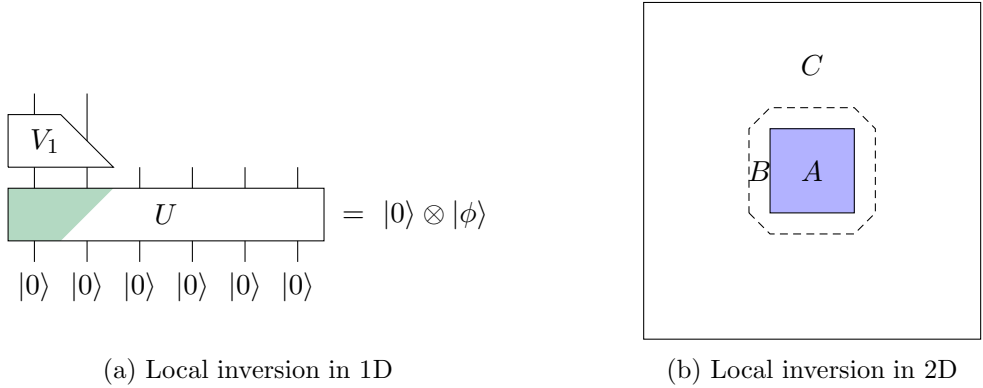

\subsection{Overview of technical challenges and new ideas}
We first give an overview of the technical challenges in developing a learning algorithm, and our new ideas to overcome these challenges. We start by defining a key concept of local inversion.

\begin{dfn}[Local inversion]\label{def:localinv}
Given a state $\ket{\psi}$ and a subset $A\subseteq [n]$ of qubits, a unitary operator $V$ is called a \emph{local inversion} of $A$ if $V$ acts on a ball of $A$ and $V\ket{\psi}=\ket{0}_A\otimes \ket{\phi}$ for some arbitrary pure state $\ket{\phi}$ on $n-|A|$ qubits.
\end{dfn}

\begin{fact}\label{fact:localinv}
    Suppose $\ket{\psi}=U\ket{0^n}$ where $U$ is a depth-$d$ circuit acting on a $k$-dimensional lattice. Then for any $A\subseteq [n]$, there exists a local inversion $V$ of $A$ satisfying:
    \begin{enumerate}
        \item $V$ is supported on $\mc B(A,d)$;
        \item $V$ is a depth $d$ circuit, whose inverse shape is contained within the lightcone $\mc L(A,d)$.
    \end{enumerate}
\end{fact}

This fact follows from \emph{undoing} the gates in the lightcone of $A$. For example, \cref{fig:1dinversion} shows a state prepared by a circuit acting on a 1D lattice, and the green region denotes the lightcone of the leftmost qubit. There exists a local inversion $V_1$ (for example, by inverting the gates in the lightcone) which disentangles the leftmost qubit. Note that the shape of $V_1$ is the inverse of the shape of the lightcone. \cref{fig:2dinversion} shows a 2D lattice, where the dotted region $AB$ equals $\mc B(A,d)$. There exists a local inversion of $A$ which is a depth-$d$ circuit acting on $AB$.

\paragraph{Learning local inversions.} Local inversions provide the basic tool for a learning algorithm because they are easy to learn. For example, consider a constant-size region $A$ in~\cref{fig:2dinversion}. We first perform quantum state tomography to learn the reduced density matrix on $AB$, and then brute force search over all depth-$d$ circuits acting on $AB$. For each circuit, we can efficiently check whether it correctly inverts the region $A$ to $\ket{0}_A$. In this way we can find possibly a set of local inversion operators of $A$.

For the remainder of~\cref{sec:learningalgorithm}, we will assume that we have access to (exact) local inversion operators. In reality, we can only learn \emph{approximate} local inversion operators. This issue is addressed in~\cref{sec:complexity}.

Our algorithm has a simple framework: (1) Quantum learning: learn local reduced density matrices; (2) Classical processing: find local inversion operators for local regions, and combine them into a circuit. Below we review the challenges in realizing this framework.



\paragraph{Challenges.} The first issue is that there are multiple local inversion operators for a given local region. A naive approach to find a circuit is the following: pick a local inversion operator for each local region, such that local inversions for neighboring regions are \emph{consistent}. Consistency demands that two quantum circuits share the same gates where they overlap, so that they can be merged together as a bigger circuit. This is precisely a constraint satisfaction problem that is hard in general in 2D and higher dimensions.

Recent work~\cite{huang2024learning} addresses this problem in 2D, by showing that one can efficiently solve such a constraint satisfaction problem in 1D, to find 1D circuits that disentangle the 2D state into many 1D stripes~\cite[Eq.~(9)]{huang2024learning}, and argue that the remaining 1D stripes are easy to learn. This approach fails at three and higher dimensions as the disentangling step still requires solving a hard constraint satisfaction problem.

The key challenge in solving the learning problem on finite dimensional lattices is to find a way to avoid any constraint satisfaction problem. In this paper we give such a way.

Here is a basic idea: by definition, applying a local inversion $V$ to a state $\ket{\psi}$ gives $V\ket{\psi}=\ket{0}_A\otimes \ket{\phi}$, where $\ket{\phi}$ is a \emph{smaller} state on $n-|A|$ qubits. Can we repeat the process by further removing qubits from $\ket{\phi}$? The issue here is that now the state has been \emph{disturbed}. In particular, the local inversion $V$ we applied may not be the ``ground truth'' which undoes the gates in the lightcone of $A$, and $V$ may just be an arbitray depth-$d$ circuit that happens to invert $A$. Therefore the state $\ket{\phi}$ suffers from a potential quantum circuit complexity blow-up: we no longer have the guarantee that $\ket{\phi}$ is prepared by a depth-$d$ circuit. Repeating this process will further increase the blow-up.


\paragraph{Key ideas.} Overcoming these obstacles requires two insights. The first is to observe that we can \emph{undo} the local inversion after applying it, so the state $\ket{\psi}$ is not disturbed. The utility of this observation is that in rewriting the state in this way, we have replaced part of the unknown state with a partial piece of known operations. The second insight is to realize that with a careful choice of local inversions based on the geometry of the lattice, these partial pieces can be layered together in a way that the \emph{backward} lightcone of the final state is covered completely by these partial pieces and does not depend on the unknown initial state. The result is the learned constant depth circuit consisting of parts of local inversions and their inverses.

\subsection{Replacement process}
As discussed above, our first key idea is the following, which at first glance seems useless: apply a local inversion, and then undo it. We formally define this as a replacement process.

\begin{dfn}[Replacement process]
\label{def:replacementprocess}
    Given a state $\ket{\psi}$ and a region $A$, define the {\em $A$-replacement process} as follows: take a local inversion $V$ of region $A$, then perform the following operations on $\ket{\psi}$:
    \begin{enumerate}
    \item apply $V$,
    \item trace out the qubits in $A$, replace each qubit with $\ket{0}$,
    \item apply $V^\dag$.
    \end{enumerate}
\end{dfn}

Note that step 2 is in fact not doing anything (since step 1 already inverted the region $A$ to $\ket{0}_A$) and is included here to help the illustration of the argument. Next, since step 2 is effectively the identity operation, step 1 and step 3 cancel with each other, and therefore the state $\ket{\psi}$ remains unchanged. This is illustrated as follows (for simplicity, here we draw local inversions as boxes without the wedges).

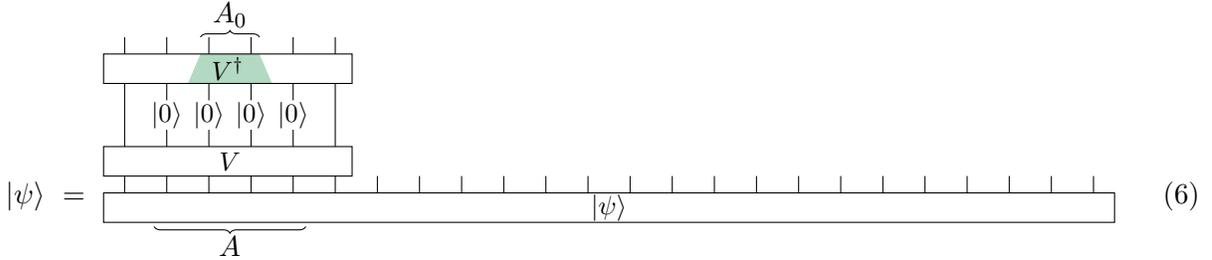
\begin{equation}\label{eq:replacementprocess}
    \ket{\psi}\,\,=\,\,\begin{tikzpicture}[baseline={(0, 0.5*\tikzlength cm-\MathAxis pt)},x=0.8*\tikzlength cm,y=0.8*\tikzlength cm]

\fill[ForestGreen!30!white] (2.3,4.0) -- (3.7,4.0) -- (4.0,3.3) -- (2.0,3.3);

        \draw[black] (0,0) rectangle node{\small{$\ket{\psi}$}} (24,0.7);
  \foreach \x in {1,2,...,24} {
        \draw[black] (\x-0.5, 0.7) -- (\x-0.5,1.1);
    }

    \draw[black] (0,1.1) rectangle (5.9,1.8);
    \node at (3,1.45){\small{$V$}};
    \foreach \x in {2,...,5} {
        \draw[black] (\x-0.5, 1.8) -- (\x-0.5,2.2);
        \node[anchor=south] at (\x-0.5, 2){\small{$\ket{0}$}};
        \draw[black] (\x-0.5, 3.3) -- (\x-0.5,2.9);
    }
    \foreach \x in {1,6} {
        \draw[black] (\x-0.5, 1.8) -- (\x-0.5,3.3);
    }
    \foreach \x in {1,...,6} {
        \draw[black] (\x-0.5, 4.0) -- (\x-0.5,4.4);
    }
    \draw[black] (0,3.3) rectangle node {\small{$V^\dag$}} (5.9,4.0) ;

    \draw [decorate,
    decoration = {brace}] (4.8,-0.1) -- (1.2,-0.1) node[pos=0.5,anchor=north]{$A$};
    \draw [decorate,
    decoration = {brace}] (2.3,4.4) -- (3.7,4.4) node[pos=0.5,anchor=south]{$A_0$};
    \end{tikzpicture}
\end{equation}

\begin{fact}\label{fact:invariant}
The state $\ket{\psi}$ is invariant under any $A$-replacement process for any region $A$.
\end{fact}

Next we give an intuitive explanation of why the replacement process may be helpful for learning. First, we informally introduce the new concept of backward lightcone (green region in \cref{eq:replacementprocess}). The formal definition is given in \cref{def:backwardlightconeformal}. 

\begin{dfn}[Backward lightcone, informal]
    The backward lightcone of a subset of qubits $S$ at the output of a quantum circuit is the minimal part of the circuit diagram that determines the reduced density matrix of $S$.
\end{dfn}
For example, in \cref{eq:replacementprocess} the green region starts from $A_0$ at the output of the circuit and keeps growing backwards (which looks like the inverse of a (forward) lightcone), until it hits a region of $\ket{0}$. There is no need to grow further because the input is completely determined. Now, note that the reduced density matrix of $A_0$ at the output of the circuit (which equals the reduced density matrix of $A_0$ in $\ket{\psi}$) is determined by its backward lightcone: start from a region of $\ket{0}$ which is larger than $A_0$, apply the green circuit, and trace out the qubits not in $A_0$. All quantum gates not in the green region are irrelevant since removing them does not affect the reduced density matrix of $A_0$.

Here is an interesting observation about the $A$-replacement process: suppose we choose $A$ to be a (sufficiently large) ball of some smaller region $A_0$ as in \cref{eq:replacementprocess}, then the backward lightcone of $A_0$ ends at the freshly initialized qubits in step 2 of \cref{def:replacementprocess}. In particular, the backward lightcone does not reach the unknown state $\ket{\psi}$, which allows us to reconstruct the reduced density matrix of $\ket{\psi}$ on $A_0$ by a \emph{known} circuit. And yet, due to the invariance of $\ket{\psi}$ (\cref{fact:invariant}) we can pretend that nothing has happened to $\ket{\psi}$ and repeat this process. In other words,

\begin{quote}
    \emph{Key observation: We have replaced part of the state with known operations, without disturbing the state.}
\end{quote}

This observation suggests an approach to learning a circuit for $\ket{\psi}$: repeatedly apply $A$-replacement processes for different small regions $A$, and hope to have the backward lightcone of more and more output qubits be contained entirely within the replacement process, and hope that eventually this holds true for all of the output qubits. If we can manage this, it means we must have generated $\ket{\psi}$ solely from the collection of replacement processes (which are constructed by quantum circuits that are known to us) and thus we can extract from them a circuit that can generate $\ket{\psi}$, simply via the backward lightcone of all output qubits. 

It is not obvious that this can work, because we need to apply replacement processes not only in parallel, but also on top of each other, since a single layer cannot cover all output qubits (e.g.~\cref{fig:1d_rec_process}). The issue is that applying replacement processes on top of each other changes the lightcone structure: for example, the backward lightcone of $A_0$ may be much larger than in \cref{eq:replacementprocess} if there are additional layers on top, because the backward lightcone starts from the output which is at the very top of the circuit.

In the next section we pin down the exact conditions for this approach to work.


\subsection{Covering schemes and reconstruction circuits}

It turns out that we can make this approach work if we can find a collection of small regions that satisfy the following conditions which we call a covering scheme. Our plan is:
\begin{enumerate}
    \item In this section we show that a covering scheme implies a learning algorithm.
    \item In \cref{sec:latticecoveringschemes} we show how to construct good covering schemes for $k$-dimensional lattices.
\end{enumerate}

\begin{dfn}[Covering scheme]
\label{def:coveringscheme}
A $(\ell,c,d)$ covering scheme is a collection of subsets of qubits $S_j^i\subseteq [n]$
\begin{equation*}
    S^1_1, S^1_2, \dots, S^1_{m_1}, S^2_1, S^2_2, \dots, S^2_{m_2}, \dots, S^\ell_1, S^\ell_2, \dots, S^\ell_{m_\ell}
\end{equation*}
which satisfy the following conditions.
\begin{enumerate}
\item The size of each $\mc B(S^i_j,d)$ is upper bounded by $c$.
\item \label{c:2} For every fixed $i$, the sets $\mc B(S_j^{i},d)$ are pairwise disjoint for $1\leq j \leq m_i$.
\item  For each qubit $v\in [n]$, there exists a $S_j^i$ such that $\mc B(\{v\},(2\ell-1)d)\subseteq S_j^i$.
\end{enumerate}
\end{dfn}

We can think of these subsets as being divided into $\ell$ different layers: in each layer $1\leq i\leq \ell$, there are subsets $S^i_1, S^i_2, \dots, S^i_{m_i}$. Condition 1 says that each of them is small (even after being enlarged by a radius of $d$). Condition 2 says that the subsets in the same layer are disjoint, even after each of them is enlarged by a radius of $d$. Condition 3 says that for each qubit, a ball around that qubit (of radius $(2\ell-1)d$) is entirely contained within some subset.

\subsubsection{Examples in 1D}
The example below shows a covering scheme in 1D with $\ell=2$ layers.
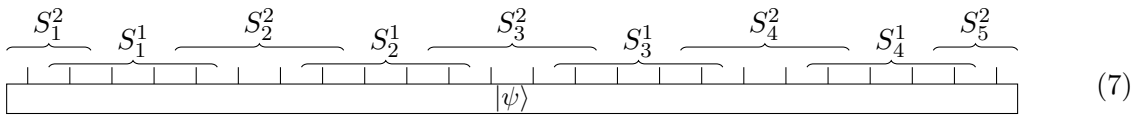
\begin{equation}
    \begin{tikzpicture}[baseline={(0, 0.5*\tikzlength cm-\MathAxis pt)},x=0.8*\tikzlength cm,y=0.8*\tikzlength cm]
        \draw[black] (0,0) rectangle node{\small{$\ket{\psi}$}} (24,0.7);
  \foreach \x in {1,2,...,24} {
        \draw[black] (\x-0.5, 0.7) -- (\x-0.5,1.1);
    }

    \foreach \t in {1,...,4} {
    \draw [decorate,
    decoration = {brace}] (1+6*\t-6,1.1) -- (5+6*\t-6,1.1)  node[pos=0.5,anchor=south]{$S_{\t}^1$};
    }
    \foreach \t in {1,...,3} {
    \draw [decorate,
    decoration = {brace}] (6*\t-6+4, 1.5) -- (6*\t+2,1.5) node[pos=0.5,anchor=south]{$S_{\pgfmathparse{\t+1}\pgfmathprintnumber{\pgfmathresult}}^2$};
    }
    \draw [decorate,
    decoration = {brace}] (0, 1.5) -- (2,1.5) node[pos=0.5,anchor=south]{$S_{1}^2$};
    \draw [decorate,
    decoration = {brace}] (22, 1.5) -- (24,1.5) node[pos=0.5,anchor=south]{$S_{5}^2$};
\end{tikzpicture}
\end{equation}

Note that in each layer, the sets have some distance between each other, because we require that they remain disjoint even after being enlarged (Condition 2). In addition, the first layer and the second layer have a lot of overlap. This ensures that any qubit must lie within the interior of some set, which implies Condition 3.

Why is this useful? Recall that our idea is to repeatedly apply replacement processes. The reason to introduce covering schemes is the following:
\begin{quote}
    \emph{Key idea: Applying replacement processes for a covering scheme allows us to reconstruct the entire state.}
\end{quote}
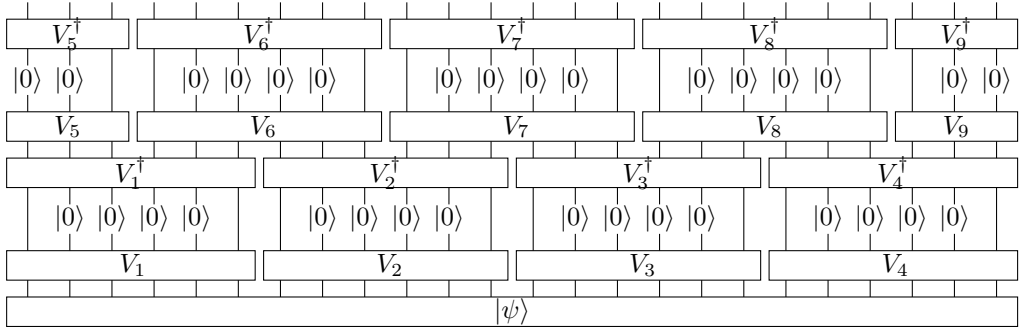
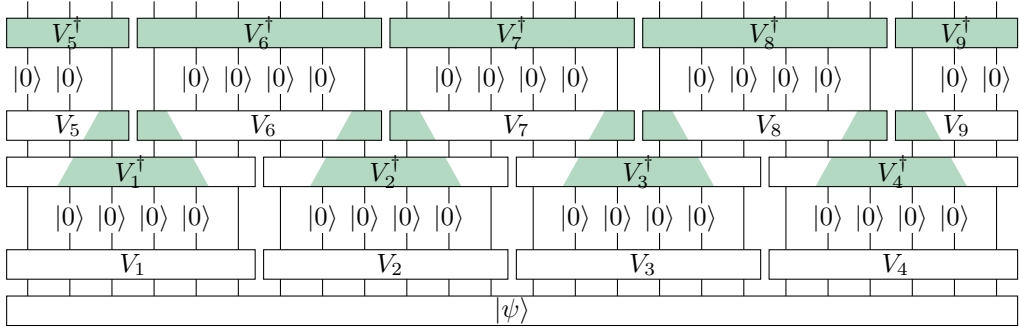
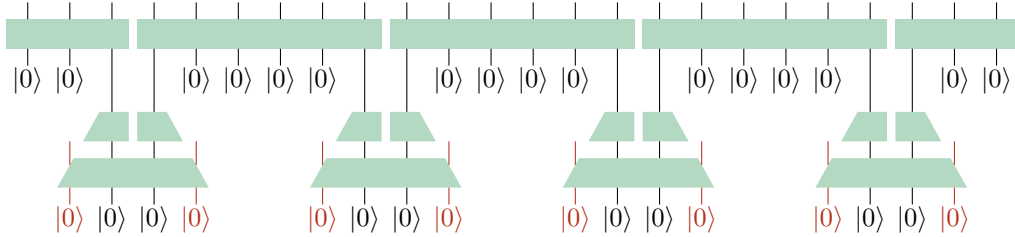
\begin{figure}[t]
    \centering
    \begin{subfigure}[b]{\textwidth}
    \centering
    \begin{tikzpicture}[baseline={(0, 0.5*\tikzlength cm-\MathAxis pt)},x=0.8*\tikzlength cm,y=0.8*\tikzlength cm]
        \draw[black] (0,0) rectangle node{\small{$\ket{\psi}$}} (24,0.7);
  \foreach \x in {1,2,...,24} {
        \draw[black] (\x-0.5, 0.7) -- (\x-0.5,1.1);
    }
\foreach \t in {1,...,4} {
    \ifthenelse{\t = 1}{\def\xmin{0};}{\def\xmin{0.1 + 6*\t-6};}
    \ifthenelse{\t = 4}{\def\xmax{24};}{\def\xmax{6*\t - 0.1};}
    \draw[black] (\xmin,1.1) rectangle (\xmax,1.8);
    \node at (\xmin + 3,1.45){\small{$V_{\t}$}};
    \foreach \x in {2,...,5} {
        \draw[black] (6*\t-6+\x-0.5, 1.8) -- (6*\t-6+\x-0.5,2.2);
        \node[anchor=south] at (6*\t-6+\x-0.5, 2){\small{$\ket{0}$}};
        \draw[black] (6*\t-6+\x-0.5, 3.3) -- (6*\t-6+\x-0.5,2.9);
    }
    \foreach \x in {1,6} {
        \draw[black] (6*\t-6+\x-0.5, 1.8) -- (6*\t-6+\x-0.5,3.3);
    }
    \foreach \x in {1,...,6} {
        \draw[black] (6*\t-6+\x-0.5, 4.0) -- (6*\t-6+\x-0.5,4.4);
    }
    \draw[black] (\xmin,3.3) rectangle node {\small{$V_{\t}^\dag$}} (\xmax,4.0);
}
\foreach \t in {1,...,3} {
    \def\xmin{0.1 + 6*\t-3};
    \def\xmax{6*\t - 0.1+3};
    \draw[black] (\xmin,4.4) rectangle (\xmax,5.1);
    \node at (\xmin + 3,4.75){\small{$V_{\pgfmathparse{\t+5}\pgfmathprintnumber{\pgfmathresult}}$}};
    \foreach \x in {2,...,5} {
        \draw[black] (6*\t-6+\x-0.5+3, 5.1) -- (6*\t-6+\x-0.5+3,5.5);
        \node[anchor=south] at (6*\t-6+\x-0.5+3, 5.3){\small{$\ket{0}$}};
        \draw[black] (6*\t-6+\x-0.5+3, 6.2) -- (6*\t-6+\x-0.5+3,6.6);
    }
    \foreach \x in {1,6} {
        \draw[black] (6*\t-6+\x-0.5+3, 5.1) -- (6*\t-6+\x-0.5+3,6.6);
    }
    \foreach \x in {1,...,6} {
        \draw[black] (6*\t-6+\x-0.5+3, 7.3) -- (6*\t-6+\x-0.5+3,7.7);
    }
    \draw[black] (\xmin,6.6) rectangle node {\small{$V_{\pgfmathparse{\t+5}\pgfmathprintnumber{\pgfmathresult}}^\dag$}} (\xmax,7.3);
}

    \draw[black] (0,4.4) rectangle node {\small{$V_5$}} (2.9,5.1);
    \foreach \x in {1,2} {
        \draw[black] (\x-0.5, 5.1) -- (\x-0.5,5.5);
        \node[anchor=south] at (\x-0.5, 5.3){\small{$\ket{0}$}};
        \draw[black] (\x-0.5, 6.2) -- (\x-0.5,6.6);
    }
    \draw[black] (2.5, 5.1) -- (2.5,6.6);
    \foreach \x in {1,...,3} {
        \draw[black] (\x-0.5, 7.3) -- (\x-0.5,7.7);
    }
    \draw[black] (0,6.6) rectangle node {\small{$V_5^\dag$}} (2.9,7.3);

    \draw[black] (21.1,4.4) rectangle node {\small{$V_9$}} (24,5.1);
    \draw[black] (21.1,6.6) rectangle node {\small{$V_9^\dag$}} (24,7.3);
    
    \foreach \x in {1,2} {
        \draw[black] (\x-0.5+22, 5.1) -- (\x-0.5+22,5.5);
        \node[anchor=south] at (\x-0.5+22, 5.3){\small{$\ket{0}$}};
        \draw[black] (\x-0.5+22, 6.2) -- (\x-0.5+22,6.6);
    }
    \draw[black] (21.5, 5.1) -- (21.5,6.6);
    \foreach \x in {1,...,3} {
        \draw[black] (\x-0.5+21, 7.3) -- (\x-0.5+21,7.7);
    }
    
\end{tikzpicture}
    \caption{Reconstruction process}
    \label{fig:1d_rec_process}
    \end{subfigure}\\[2mm]
    \begin{subfigure}[b]{\textwidth}
    \centering
    \begin{tikzpicture}[baseline={(0, 0.5*\tikzlength cm-\MathAxis pt)},x=0.8*\tikzlength cm,y=0.8*\tikzlength cm]

\foreach \t in {1,...,3} {
\def\xmin{0.1 + 6*\t-3};
    \def\xmax{6*\t - 0.1+3};
    \fill[ForestGreen!30!white] (\xmin,6.6) rectangle (\xmax,7.3);
}
\fill[ForestGreen!30!white] (21.1,6.6) rectangle (24,7.3);
\fill[ForestGreen!30!white] (0,6.6) rectangle (2.9,7.3);
\foreach \t in {1,...,4} {
\fill[ForestGreen!30!white] (6*\t-6+3-0.8,5.1) -- (6*\t-6+3-0.1,5.1) -- (6*\t-6+3-0.1,4.4) -- (6*\t-6+3-1.1889,4.4);
\fill[ForestGreen!30!white] (6*\t-6+3+0.8,5.1) -- (6*\t-6+3+0.1,5.1) -- (6*\t-6+3+0.1,4.4) -- (6*\t-6+3+1.1889,4.4);
\fill[ForestGreen!30!white] (6*\t-6+3-1.4111,4.0) -- (6*\t-6+3-1.8,3.3) -- (6*\t-6+3+1.8,3.3) -- (6*\t-6+3+1.4111,4.0);
}

        \draw[black] (0,0) rectangle node{\small{$\ket{\psi}$}} (24,0.7);
  \foreach \x in {1,2,...,24} {
        \draw[black] (\x-0.5, 0.7) -- (\x-0.5,1.1);
    }

\foreach \t in {1,...,4} {
    \ifthenelse{\t = 1}{\def\xmin{0};}{\def\xmin{0.1 + 6*\t-6};}
    \ifthenelse{\t = 4}{\def\xmax{24};}{\def\xmax{6*\t - 0.1};}
    \draw[black] (\xmin,1.1) rectangle (\xmax,1.8);
    \node at (\xmin + 3,1.45){\small{$V_{\t}$}};
    \foreach \x in {2,...,5} {
        \draw[black] (6*\t-6+\x-0.5, 1.8) -- (6*\t-6+\x-0.5,2.2);
        \node[anchor=south] at (6*\t-6+\x-0.5, 2){\small{$\ket{0}$}};
        \draw[black] (6*\t-6+\x-0.5, 3.3) -- (6*\t-6+\x-0.5,2.9);
    }
    \foreach \x in {1,6} {
        \draw[black] (6*\t-6+\x-0.5, 1.8) -- (6*\t-6+\x-0.5,3.3);
    }
    \foreach \x in {1,...,6} {
        \draw[black] (6*\t-6+\x-0.5, 4.0) -- (6*\t-6+\x-0.5,4.4);
    }
    \draw[black] (\xmin,3.3) rectangle node {\small{$V_{\t}^\dag$}} (\xmax,4.0);
}

\foreach \t in {1,...,3} {
    \def\xmin{0.1 + 6*\t-3};
    \def\xmax{6*\t - 0.1+3};
    \draw[black] (\xmin,4.4) rectangle (\xmax,5.1);
    \node at (\xmin + 3,4.75){\small{$V_{\pgfmathparse{\t+5}\pgfmathprintnumber{\pgfmathresult}}$}};
    \foreach \x in {2,...,5} {
        \draw[black] (6*\t-6+\x-0.5+3, 5.1) -- (6*\t-6+\x-0.5+3,5.5);
        \node[anchor=south] at (6*\t-6+\x-0.5+3, 5.3){\small{$\ket{0}$}};
        \draw[black] (6*\t-6+\x-0.5+3, 6.2) -- (6*\t-6+\x-0.5+3,6.6);
    }
    \foreach \x in {1,6} {
        \draw[black] (6*\t-6+\x-0.5+3, 5.1) -- (6*\t-6+\x-0.5+3,6.6);
    }
    \foreach \x in {1,...,6} {
        \draw[black] (6*\t-6+\x-0.5+3, 7.3) -- (6*\t-6+\x-0.5+3,7.7);
    }
    \draw[black] (\xmin,6.6) rectangle node {\small{$V_{\pgfmathparse{\t+5}\pgfmathprintnumber{\pgfmathresult}}^\dag$}} (\xmax,7.3);
}

    \draw[black] (0,4.4) rectangle node {\small{$V_5$}} (2.9,5.1);
    \foreach \x in {1,2} {
        \draw[black] (\x-0.5, 5.1) -- (\x-0.5,5.5);
        \node[anchor=south] at (\x-0.5, 5.3){\small{$\ket{0}$}};
        \draw[black] (\x-0.5, 6.2) -- (\x-0.5,6.6);
    }
    \draw[black] (2.5, 5.1) -- (2.5,6.6);
    \foreach \x in {1,...,3} {
        \draw[black] (\x-0.5, 7.3) -- (\x-0.5,7.7);
    }
    \draw[black] (0,6.6) rectangle node {\small{$V_5^\dag$}} (2.9,7.3);

    \draw[black] (21.1,4.4) rectangle node {\small{$V_9$}} (24,5.1);
    \draw[black] (21.1,6.6) rectangle node {\small{$V_9^\dag$}} (24,7.3);
    
    \foreach \x in {1,2} {
        \draw[black] (\x-0.5+22, 5.1) -- (\x-0.5+22,5.5);
        \node[anchor=south] at (\x-0.5+22, 5.3){\small{$\ket{0}$}};
        \draw[black] (\x-0.5+22, 6.2) -- (\x-0.5+22,6.6);
    }
    \draw[black] (21.5, 5.1) -- (21.5,6.6);
    \foreach \x in {1,...,3} {
        \draw[black] (\x-0.5+21, 7.3) -- (\x-0.5+21,7.7);
    }
    
\end{tikzpicture}
    \caption{Backward lightcone}
    \label{fig:1d_back_lightcone}
    \end{subfigure}\\[2mm]
    \begin{subfigure}[b]{\textwidth}
    \centering
    \begin{tikzpicture}[baseline={(0, 1.8*\tikzlength cm-\MathAxis pt)},x=0.8*\tikzlength cm,y=0.8*\tikzlength cm]

\foreach \t in {1,...,3} {
\def\xmin{0.1 + 6*\t-3};
    \def\xmax{6*\t - 0.1+3};
    \fill[ForestGreen!30!white] (\xmin,6.6) rectangle (\xmax,7.3);
}
\fill[ForestGreen!30!white] (21.1,6.6) rectangle (24,7.3);
\fill[ForestGreen!30!white] (0,6.6) rectangle (2.9,7.3);
\foreach \t in {1,...,4} {
\fill[ForestGreen!30!white] (6*\t-6+3-0.8,5.1) -- (6*\t-6+3-0.1,5.1) -- (6*\t-6+3-0.1,4.4) -- (6*\t-6+3-1.1889,4.4);
\fill[ForestGreen!30!white] (6*\t-6+3+0.8,5.1) -- (6*\t-6+3+0.1,5.1) -- (6*\t-6+3+0.1,4.4) -- (6*\t-6+3+1.1889,4.4);
\fill[ForestGreen!30!white] (6*\t-6+3-1.4111,4.0) -- (6*\t-6+3-1.8,3.3) -- (6*\t-6+3+1.8,3.3) -- (6*\t-6+3+1.4111,4.0);
}


\foreach \t in {1,...,4} {
    \foreach \x in {2,5} {
        \node[anchor=south] at (6*\t-6+\x-0.5, 2){{\color{BrickRed} \small{$\ket{0}$}}};
        \draw[BrickRed] (6*\t-6+\x-0.5, 3.3) -- (6*\t-6+\x-0.5,2.9);
    }
    \foreach \x in {3,4} {
        \node[anchor=south] at (6*\t-6+\x-0.5, 2){\small{$\ket{0}$}};
        \draw[black] (6*\t-6+\x-0.5, 3.3) -- (6*\t-6+\x-0.5,2.9);
    }
    \foreach \x in {2,5} {
        \draw[BrickRed] (6*\t-6+\x-0.5, 3.84039) -- (6*\t-6+\x-0.5,4.4);
    }
    \foreach \x in {3,4} {
        \draw[black] (6*\t-6+\x-0.5, 4.0) -- (6*\t-6+\x-0.5,4.4);
    }
}

\foreach \t in {1,...,3} {
    \def\xmin{0.1 + 6*\t-3};
    \def\xmax{6*\t - 0.1+3};
    \foreach \x in {2,...,5} {
        \node[anchor=south] at (6*\t-6+\x-0.5+3, 5.3){\small{$\ket{0}$}};
        \draw[black] (6*\t-6+\x-0.5+3, 6.2) -- (6*\t-6+\x-0.5+3,6.6);
    }
    \foreach \x in {1,6} {
        \draw[black] (6*\t-6+\x-0.5+3, 5.1) -- (6*\t-6+\x-0.5+3,6.6);
    }
    \foreach \x in {1,...,6} {
        \draw[black] (6*\t-6+\x-0.5+3, 7.3) -- (6*\t-6+\x-0.5+3,7.7);
    }
}

    \foreach \x in {1,2} {
        \node[anchor=south] at (\x-0.5, 5.3){\small{$\ket{0}$}};
        \draw[black] (\x-0.5, 6.2) -- (\x-0.5,6.6);
    }
    \draw[black] (2.5, 5.1) -- (2.5,6.6);
    \foreach \x in {1,...,3} {
        \draw[black] (\x-0.5, 7.3) -- (\x-0.5,7.7);
    }

    
    \foreach \x in {1,2} {
        \node[anchor=south] at (\x-0.5+22, 5.3){\small{$\ket{0}$}};
        \draw[black] (\x-0.5+22, 6.2) -- (\x-0.5+22,6.6);
    }
    \draw[black] (21.5, 5.1) -- (21.5,6.6);
    \foreach \x in {1,...,3} {
        \draw[black] (\x-0.5+21, 7.3) -- (\x-0.5+21,7.7);
    }
    
\end{tikzpicture}
    \caption{Reconstructed circuit}
    \label{fig:1d_rec_circuit}
    \end{subfigure}
    \caption{Illustration of the learning algorithm in 1D.}
    \label{fig:1dreconstruction}
\end{figure}

We first illustrate this idea in 1D in \cref{fig:1dreconstruction}, and then give a formal argument in~\cref{sec:reconstruction}.

\begin{itemize}
    \item \textbf{Reconstruction process.} Suppose we apply $S_1^1$-replacement process, which is supported on $\mc B(S_1^1,d)$ and looks exactly the same as~\cref{eq:replacementprocess}. Note that all replacement processes corresponding to the first layer in the covering scheme can be implemented in parallel, due to Condition 2 in \cref{def:coveringscheme}. Next, we apply all replacement processes corresponding to the second layer in the covering scheme. We call the resulting diagram a reconstruction process, shown in \cref{fig:1d_rec_process}.
    \item \textbf{Backward lightcone.} Now we construct the backward lightcone for all output wires of the reconstruction process (\cref{fig:1d_back_lightcone}). The way to do this is to color all gates at the top layer in green, and then ``spread'' the green color backwards, until hitting a region of $\ket{0}$. Note that some of the spreading stops at the second layer of $\ket{0}$, but inevitably some of the spreading goes beyond the second layer and enters the first layer. However, all of the spreading stops at the first layer of $\ket{0}$ and never touches the bottom unknown state $\ket{\psi}$.
    \item \textbf{Reconstructed circuit.} By \cref{fact:invariant}, the state $\ket{\psi}$ remains invariant under the reconstruction process, and therefore the output state at the top layer equals $\ket{\psi}$. The backward lightcone consists entirely of known quantum circuits, and therefore it gives a reconstructed circuit that prepares $\ket{\psi}$ (\cref{fig:1d_rec_circuit}). This circuit has the following features: it has depth $3d$ and acts on the all-$\ket{0}$ input state, where we view the red qubits as ancilla qubits. After the circuit is applied, the entire state equals $\ket{\psi}\otimes\ket{\mathrm{junk}}$, where the wires at the top correspond to $\ket{\psi}$, and the red wires (ancilla qubits) correspond to $\ket{\mathrm{junk}}$.
\end{itemize}

\subsubsection{Reconstruction theorem}
In this section we give a rigorous argument for how to reconstruct the state in general. We start with a formal definiton of the reconstruction process.

\label{sec:reconstruction}
\begin{dfn}[Reconstruction process]
\label{def:reconstructionprocess}
The reconstruction process for $\ket{\psi}$ is defined as follows. Let $V_{i,j}$ be a local inversion for $S_j^i$. By \cref{fact:localinv}, $V_{i,j}$ is a depth-$d$ circuit acting on $\mc B(S_j^i,d)$. The reconstruction process is defined as follows: 
\begin{quote}
    For each $1\leq i\leq \ell$, apply the $S_j^i$-replacement process using $V_{i,j}$ for all $j$ in parallel.
\end{quote}
More specifically, for each fixed $i$ we do the following: 
\begin{enumerate}
    \item Apply $V_{i,j}$ for all $1\leq j\leq m_i$ in parallel;
    \item Trace out all qubits in $S_j^i$ for each $1\leq j\leq m_i$ and replace with $\ket{0}$;
    \item Apply $V_{i,j}^\dag$ for all $1\leq j\leq m_i$ in parallel.
\end{enumerate}
\end{dfn}

As shown in the 1D example in \cref{fig:1dreconstruction}, a two-step argument is used to show that a circuit for preparing the unknown state $\ket{\psi}$ can be extracted from the reconstruction process:
\begin{enumerate}
    \item The state $\ket{\psi}$ is invariant under each $S_j^i$-replacement process (\cref{fact:invariant}), and therefore is invariant under the entire reconstruction process.
    \item We can reconstruct the output state of the reconstruction process (which equals $\ket{\psi}$) via its backward lightcone, because the backward lightcone consists of known quantum circuits and in particular does not touch the unknown input state $\ket{\psi}$ at the bottom.
\end{enumerate}

Below we elaborate on the second point. At this point a more precise definition of the backward lightcone is needed. Note that this definition needs to be applicable to slightly more general quantum circuits with reset gates.

\begin{dfn}[Backward lightcone]\label{def:backwardlightconeformal}
    Let $\ket{\phi}=W\ket{0^n}$ where $W$ is a quantum circuit consists of 2-qubit unitary gates and 1-qubit reset gates (a reset gate traces out the input and initializes a $\ket{0}$ state). Let $A\subseteq [n]$ be a subset of qubits. The backward lightcone of $A$ is a circuit diagram which is part of the circuit diagram of $W$, defined as the collection of green gates acting on all-0 inputs, constructed as follows:
    \begin{enumerate}
        \item Color the output wires of $W$ corresponding to $A$ in blue.
        \item Repeat the following process until no changes happen:\\
        If there exists a 2-qubit gate $G$ with a blue wire on top, color this gate in green. Moreover, consider each of the two wires at the bottom of $G$. If it is not connected to a reset gate, color the wire in blue.
    \end{enumerate}
In other words, the backward lightcone consists of all quantum gates that could influence the reduced density matrix of $\ketbra{\phi}$ on $A$. The running time to construct the backward lightcone is linear in the size of $W$.
\end{dfn}

The process to construct the backward lightcone can be viewed as spreading the green color from top to bottom, as shown in~\cref{fig:1d_back_lightcone}. The desired property for the backward lightcone, when applied to a reconstruction process, is that it stops entirely at intermediate reset gates and does not reach the bottom.

To help further illustrate this concept, we introduce a dual thought experiment. A different way of formulating our desired property is that we do not want the influence of the bottom input state to reach any top wire. The propagation of the influence can be viewed as forward lightcone spreading of the input state, that is, spreading the white color in~\cref{fig:1d_back_lightcone} from bottom to top. Note that indeed, the white color stops entirely at intermediate reset gates and does not reach the top.

Now we are ready to prove that the desired property is guaranteed by the covering scheme.

\begin{thm}[Covering scheme implies learning algorithm]
\label{thm:coveringimplieslearning}

Suppose $\ket{\psi}=U\ket{0^n}$ where $U$ is a depth-$d$ circuit acting on a $k$-dimensional lattice. Let $(\ell,c,d)$ be a covering scheme for the $k$-dimensional lattice, and let $W$ be the reconstruction process for this covering scheme, which satisfies $W\ket{\psi}=\ket{\psi}$. Then the backward lightcone of all output wires in $W\ket{\psi}$ is a depth-$(2\ell-1)d$ circuit $C$ which is part of $W$. 

The circuit $C$ can be used to prepare $\ket{\psi}$ in the following sense: it acts on $n$ qubits as well as $m=O(n)$ ancilla qubits, and
\begin{equation}\label{eq:tensorproduct}
    C\ket{0^n}\otimes \ket{0^m}=\ket{\psi}\otimes\ket{\mathrm{junk}}.
\end{equation}
\end{thm}

\begin{proof}
    To prove that the backward lightcone of all output wires in $W\ket{\psi}$ is part of $W$ and does not reach the input state $\ket{\psi}$, it suffices to show that the backward lightcone of each individual output wire in $W\ket{\psi}$ must stop at some regions of reset gates in some replacement processes.

    Keeping track of the backward lightcone can be tricky: the backward spreading process can stop at various different layers such as in \cref{fig:1d_back_lightcone}. However, here we give a simple and pessimistic argument which suffices for our purpose.

    Focus on a single output wire $v\in[n]$ of $W\ket{\psi}$ and imagine its backward spreading process. Part of the spreading may stop early at some reset gates, while other parts of the spreading may continue further. Instead of working with $W$, suppose we consider a new quantum circuit $W'$ where all reset gates in $W$ is replaced by the identity gate, with one exception: according to Condition 3 of \cref{def:coveringscheme}, there exists a $S_j^i$ such that $\mc B(\{v\},(2\ell-1)d)\subseteq S_j^i$; we keep the reset gates in $W'$ that correspond to the $S_j^i$-replacement process.
    
    Clearly, the backward spreading process in $W$ is entirely contained within the backward spreading process in $W'$, since $W'$ removed some reset gates relative to $W$ which can only help the spreading. Now, observe the following: \emph{the backward spreading process in $W'$ must stop at the reset gates in the $S_j^i$-replacement process.} This is because $S_j^i$ contains a ball around $v$ with radius $(2\ell-1)d$. By the time the spreading process reaches those reset gates, the process cannot spread further than a distance of $(2\ell-1)d$ and therefore is entirely covered by those reset gates. The number $(2\ell-1)d$ comes from a worst case estimate: suppose $S_j^i$ is within the very bottom layer of the covering scheme, then the spreading process has gone through the top $\ell-1$ layers of depth $2d$ as well as $V_{i,j}^\dag$ of depth $d$, with a total depth of $(2\ell-1)d$. This implies that the backward spreading process of $v$ in $W$ must be contained within $W$: if it cannot reach the bottom even in $W'$, then it also cannot reach the bottom in $W$.

    This concludes the proof that $C$ is a (at most) depth-$(2\ell-1)d$ circuit and is part of $W$. Finally, note that $C$ is a unitary quantum circuit which outputs a pure state. In addition, the reduced density matrix on the $n$ output wires equals $\ketbra{\psi}$. This implies that the system and ancilla must be tensor product pure states as in~\cref{eq:tensorproduct}.
\end{proof}



To conclude this section we give a recap about the role of different parameters of a $(\ell,c,d)$ covering scheme in the learning algorithm.

\begin{itemize}
    \item $d$ is the promised depth of the unknown circuit that prepares $\ket{\psi}$.
    \item $\ell$ determines the depth of the learned circuit, which equals $(2\ell-1)d$.
    \item $c$ determines the running time of the learning algorithm: the algorithm needs to find local inversions acting on a region of size $c$, which takes time exponential in $c$ and $d$.
\end{itemize}
Clearly, these parameters play an important role and are in tension with each other.

\subsection{Good covering schemes}
\label{sec:latticecoveringschemes}

\input{fig_2d}

In this section we construct good covering schemes for $k$-dimensional lattices.

\begin{thm}[Lattice covering scheme]
\label{thm:latticecovering}
    For $k$-dimensional lattice there exists a $\left(k+1,c,d\right)$ covering scheme for any integer $d>0$, where
    \begin{equation}
        c\leq \left((8k^2+14k+2)d\right)^k.
    \end{equation}
\end{thm}

Before giving the proof, we first discuss an example in 2D (\cref{fig:2dcircuit}). Our starting point is a new concept called lattice coloring.

\begin{dfn}[Lattice coloring]
\label{def:latticecoloring}
    A $(w,c,R)$ lattice coloring for the $k$-dimensional lattice is defined as follows. Suppose the lattice is divided into disjoint subsets where each subset is connected. Each subset is assigned a color. We demand the following properties:
    \begin{enumerate}
        \item There are $w$ different colors.
        \item Each subset has size at most $c$.
        \item Two different subsets with the same color must have distance at least $R$.
    \end{enumerate}
\end{dfn}

\cref{fig:2dlatticecoloring} shows a $(3,16R^2,R)$ coloring for the 2D lattice ($16R^2$ is an overestimate). Here $R$ is a free parameter. Similar coloring schemes have been widely used in the physics literature (e.g.~\cite{Brandão2019Finite}). Now, we use this coloring to construct a covering scheme for the 2D lattice.

\begin{lem}
    For 2D lattice there exists a $\left(3,3844 d^2,d\right)$ covering scheme for any integer $d>0$.
\end{lem}
This covering scheme is shown in \cref{fig:2dlayer1,fig:2dlayer2,fig:2dlayer3}. Below we explain it in detail.

The idea is the following: each color in the lattice coloring corresponds to a layer in the covering scheme (so $\ell=3$). We choose $R$ to be large enough, and then
\begin{itemize}
    \item Fix a color in the lattice coloring (say red regions in \cref{fig:2dlatticecoloring}), then for each small colored subset $A$ in \cref{fig:2dlatticecoloring}, we assign a subset $S=\mc B(A,5d)$ to the covering scheme (red regions in \cref{fig:2dlayer2}).
\end{itemize}
Now, we choose $R$ to make sure that Condition 2 in \cref{def:coveringscheme} is satisfied. Say we consider two red regions $A$ and $B$ in \cref{fig:2dlatticecoloring} that are separated by distance $R$. The corresponding subsets in the covering scheme are $S_1=\mc B(A,5d)$ and $S_2=\mc B(A,5d)$. Condition 2 in \cref{def:coveringscheme} demands that $S_1$ and $S_2$ are disjoint even after being enlarged further by distance $d$ (dashed regions in \cref{fig:2dlayer2}). That is, $\mc B(A,6d)$ and $\mc B(B,6d)$ must be disjoint. We therefore choose $R=13d$ to ensure this property.

Next, we show that Condition 3 in \cref{def:coveringscheme} is satisfied. Take any qubit $v\in[n]$, then it must lie within \emph{some} colored region in \cref{fig:2dlatticecoloring}. Suppose the region is red and call it $A$. Then $\mc B(\{v\},5d)\subseteq\mc B(A,5d)$ which is one of the subsets in the covering scheme in \cref{fig:2dlayer2}.

Finally, note that each small colored region in \cref{fig:2dlatticecoloring} is contained with in a box of $4R\times 4R$. The length scale of each subset in the covering scheme is at most $4R+2\times 5d=62d$. Therefore, all subsets in the 2D covering scheme has size at most $62d\times 62d=3844 d^2$.

\begin{proof}[Proof of~\cref{thm:latticecovering}]
The proof essentially repeats the above example. We first quote a coloring scheme for the $k$-dimensional lattice (see e.g.~\cite[below Definition 17]{huang2024learning}).

\begin{lem}[$k$-dimensional lattice coloring]
    For $k$-dimensional lattice there exists a $\left(k+1,(2kR)^k,R\right)$ lattice coloring. Each small colored region is contained in a $k$-dimensional box of length scale $2kR$.
\end{lem}

Each color in the lattice coloring corresponds to a layer in the covering scheme (so $\ell=k+1$). We choose $R$ to be large enough, and then for each small colored subset $A$ in the lattice coloring, we assign a subset $S=\mc B(A,(2k+1)d)$ to the corresponding layer of the covering scheme.
\begin{itemize}
    \item The depth of learned circuit is $(2\ell-1)d=(2k+1)d$.
    \item To ensure Condition 2 in \cref{def:coveringscheme} is satisfied, we choose $R=2\times (2k+2)d+d=(4k+5)d$.
    \item The length scale of each subset in the covering scheme is at most $2kR + 2\times (2k+1)d = (8k^2+14k+2)d$. Therefore, $c$ is at most $\left((8k^2+14k+2)d\right)^k$.
\end{itemize}
\end{proof}

\begin{remark}\label{remark:covering}
    In the proof of \cref{thm:latticecovering} we have chosen $R=(4k+5)d$ which leads to the stated scaling of $c$. Note that choosing $R$ to be any number larger than that also gives a valid covering scheme, with a larger $c$. This is useful for the discussions in \cref{sec:ancilla}.
\end{remark}

\section{Detailed analysis}
\label{sec:complexity}
In this section we give a detailed analysis which leads to the following main result.

\begin{thm}[Main result]
\label{thm:maindetailed}
    There is an algorithm that, given copies of an unknown state $\ket{\psi}$, with the promise that $\ket{\psi}=U\ket{0^n}$ where $U$ is an unknown depth-$d$ circuit acting on a $k$-dimensional lattice (using arbitrary 2-qubit gates), outputs a depth-$(2k+1)d$ circuit $W$ that prepares $\ket{\psi}$ up to $\varepsilon$ trace distance, with success probability $1-\delta$. The algorithm uses $M$ copies of $\ket{\psi}$ and runs in time $T$, where
    \begin{equation}
        M=\frac{n^4\cdot 2^{O(c)}}{\varepsilon^4} \log\frac{n}{\delta},\quad T=\frac{n^4\cdot 2^{O(c)}}{\varepsilon^4} \log\frac{n}{\delta}+\left(\frac{nkd\cdot c}{\varepsilon}\right)^{O(d\cdot c)}.
    \end{equation}
Here, $c=O((3k)^{k+2} d)^k$, and $W$ uses $r\cdot n$ ancilla qubits where $r>0$ can be chosen to be an arbitrarily small constant.
\end{thm}



\subsection{The reconstruction process is robust}
We first show that the final error is controlled, if we replace local inversions in the reconstruction process with approximate local inversions.

\begin{dfn}[$\varepsilon$-approximate local inversion]
    Let $V$ be a unitary acting on $AB$, and denote the remaining system as $C$ (\cref{fig:2dinversion}). Let $\ket{\psi}$ be a state defined on $ABC$. $V$ is an $\varepsilon$-approximate local inversion of the region $A$ for $\ket{\psi}$ if
    \begin{equation}
        \bra{0}_A \Tr_{BC}\left(V\ketbra{\psi}V^\dag\right)\ket{0}_A \geq 1-\varepsilon.
    \end{equation}
\end{dfn}

\begin{lem}
\label{lemma:robustreplacement}
    Let $V$ be an $\varepsilon$-approximate local inversion of the region $A$ for $\ket{\psi}$. The corresponding $A$-replacement process is given by
    \begin{equation}\label{eq:approxreplacement}
        \mc V^\dag \circ \mc R_A \circ \mc V,
    \end{equation}
    where calligraphic letters denote channels, and $\mc R_A$ is the reset channel acting on $A$ which traces out the input and prepares $\ketbra{0}_A$. Then we have
    \begin{equation}
        \left\|\mc V^\dag \circ \mc R_A \circ \mc V(\ketbra{\psi})-\ketbra{\psi}\right\|_1\leq 4\sqrt{\varepsilon}.
    \end{equation}
\end{lem}
\begin{proof}
    First, note that due to the unitary invariance of the trace distance,
    \begin{equation}
    \begin{aligned}
         \left\|\mc V^\dag \circ \mc R_A \circ \mc V(\ketbra{\psi})-\ketbra{\psi}\right\|_1 &= \left\|\mc V^\dag \circ \mc R_A \circ \mc V(\ketbra{\psi})-\mc V^\dag\circ\mc V(\ketbra{\psi})\right\|_1\\
         &=\left\| \mc R_A \circ \mc V(\ketbra{\psi})-\mc V(\ketbra{\psi})\right\|_1.
    \end{aligned}
    \end{equation}
    Next, we can write $V\ket{\psi}$ as
    \begin{equation}
        V\ket{\psi}=\sqrt{1-\varepsilon'}\ket{0}_A\ket{\phi}_{BC}+\sqrt{\varepsilon'}\ket{\mathrm{else}}_{ABC},
    \end{equation}
    where $\varepsilon'\leq \varepsilon$ and $\bra{0}_A \ket{\mathrm{else}}_{ABC}=0$. Using the relationship between fidelity and trace distance,
    \begin{equation}
        \left\|\mc V(\ketbra{\psi})-\ketbra{0}_A\otimes\ketbra{\phi}_{BC}\right\|_1=2\sqrt{1-\left|\bra{\psi}V^\dag\ket{0}_A\ket{\phi}_{BC}\right|^2}=2\sqrt{\varepsilon'}.
    \end{equation}
    Finally,
    \begin{equation}
        \begin{aligned}
            &\left\| \mc R_A \circ \mc V(\ketbra{\psi})-\mc V(\ketbra{\psi})\right\|_1\\
            &\leq \left\| \mc R_A \circ \mc V(\ketbra{\psi})-\ketbra{0}_A\otimes\ketbra{\phi}_{BC}\right\|_1+\left\|\ketbra{0}_A\otimes\ketbra{\phi}_{BC}-\mc V(\ketbra{\psi})\right\|_1\\
            &=\left\| \mc R_A \circ \mc V(\ketbra{\psi})-\mc R_A\left(\ketbra{0}_A\otimes\ketbra{\phi}_{BC}\right)\right\|_1+\left\|\ketbra{0}_A\otimes\ketbra{\phi}_{BC}-\mc V(\ketbra{\psi})\right\|_1\\
            &\leq \left\| \mc V(\ketbra{\psi})-\ketbra{0}_A\otimes\ketbra{\phi}_{BC}\right\|_1+\left\|\ketbra{0}_A\otimes\ketbra{\phi}_{BC}-\mc V(\ketbra{\psi})\right\|_1\\
            &\leq 4\sqrt{\varepsilon}.
        \end{aligned}
    \end{equation}
Here, the second line is by triangle inequality; the third line is because $\ketbra{0}_A\otimes\ketbra{\phi}_{BC}$ is invariant under $\mc R_A$; the fourth line is by data-processing inequality.
\end{proof}

\begin{lem}\label{lemma:robustreconstructionprocess}
    Let $\Phi_1,\Phi_2,\dots,\Phi_T$ be arbitrary replacement processes using $\varepsilon$-approximate local inversions of $\ket{\psi}$, where each $\Phi_i$ has the form of~\cref{eq:approxreplacement} where $V$ is an $\varepsilon$-approximate local inversion for some arbitrary region. Then
    \begin{equation}
        \left\|\Phi_T\circ \cdots\circ \Phi_2\circ \Phi_1(\ketbra{\psi})-\ketbra{\psi}\right\|_1\leq 4T\sqrt{\varepsilon}.
    \end{equation}
\end{lem}
\begin{proof}
Note that
    \begin{equation}
        \begin{aligned}
            &\left\|\Phi_T\circ \cdots\circ \Phi_2\circ \Phi_1(\ketbra{\psi})-\ketbra{\psi}\right\|_1\\
            &\leq \left\|\Phi_T\circ \cdots\circ \Phi_2\circ \Phi_1(\ketbra{\psi})-\Phi_T(\ketbra{\psi})\right\|_1+\left\|\Phi_T(\ketbra{\psi})-\ketbra{\psi}\right\|_1\\
            &\leq \left\|\Phi_{T-1} \cdots\circ \Phi_2\circ \Phi_1(\ketbra{\psi})-\ketbra{\psi}\right\|_1+4\sqrt{\varepsilon},
        \end{aligned}
    \end{equation}
where the second line is by triangle inequality, the third line is by data-processing inequality and \cref{lemma:robustreplacement}. The claim follows from induction.
\end{proof}

\subsection{Analysis of the learning algorithm}
\label{sec:analysis}

Next we put everything together and give a detailed analysis of the learning algorithm. The algorithm is given copies of an unknown quantum state $\ket{\psi}$, with the promise that $\ket{\psi}=U\ket{0^n}$ where $U$ is a depth-$d$ circuit acting on a $k$-dimensional lattice.

\paragraph{Step 1: build a covering scheme.} We build a $\left(k+1,c,d\right)$ covering scheme (here $c=O(k^2 d)^k$ or larger, depending on the choice of $R$, see~\cref{remark:covering}) for the $k$-dimensional lattice according to \cref{thm:latticecovering}. This covering scheme is divided into $\ell=k+1$ layers. Note that there are less than $n$ subsets in total.

\paragraph{Step 2: learn reduced density matrices.} To find local inversions for the covering scheme it suffices to learn reduced density matrices on a radius-$d$ ball around each subset (dashed regions in~\cref{fig:2dlayer1,fig:2dlayer2,fig:2dlayer3}). Each reduced density matrix has size which has the same scaling as $c$ and there are less than $n$ of them. 

Suppose we would like to learn all of them within $\varepsilon_1$ trace distance, with $\delta$ failure probability. Here we use a simple existing result (see e.g.~\cite[Lemma 23]{huang2024learning}): for each copy of $\ket{\psi}$, we measure each qubit in a random Pauli basis. The collection of measurement outcomes is known as classical shadows or randomized measurement dataset~\cite{Huang2020Predicting,Elben2023randomized}. The desired result can be achieved by processing this dataset, which uses 
\begin{equation}
    M=\frac{2^{O(c)}}{\varepsilon_1^2} \log\frac{n}{\delta}
\end{equation}
copies of $\ket{\psi}$. The running time to process this dataset scales similarly. From now on everything is classical processing on the learned reduced density matrices, and we assume that all learned reduced density matrices are within $\varepsilon_1$ error (this happens with probability at least $1-\delta$).

\paragraph{Step 3: find approximate local inversions.} 

For each reduced density matrix, suppose it looks like the region $AB$ in \cref{fig:2dinversion}. We will brute force search over an $\varepsilon_1$-net on depth-$d$ circuits acting on $AB$ to find a depth-$d$ circuit that approximately inverts the region $A$. In this way we can find a $3\varepsilon_1$-approximate local inversion of $A$ for $\ket{\psi}$ (see e.g.~\cite[Proof of first claim of Theorem 9]{huang2024learning} for a proof). The running time scales as the size of the $\varepsilon_1$-net, which equals
\begin{equation}
    \left(\frac{kd\cdot c}{\varepsilon_1}\right)^{O(d\cdot c)}.
\end{equation}
See e.g.~\cite[Lemma 19]{huang2024learning} for a proof.

Note that this can be significantly improved if we assume a discrete gate set, which we do not discuss here.

\paragraph{Step 4: find approximate reconstruction circuit.} We have found a $3\varepsilon_1$-approximate local inversion for every subset in the covering scheme. Construct a reconstruction process (\cref{def:reconstructionprocess}) using these approximate local inversions, denote as $\Phi$. By \cref{lemma:robustreconstructionprocess}, we have
\begin{equation}
    \left\|\Phi(\ketbra{\psi})-\ketbra{\psi}\right\|_1\leq 4n\sqrt{3\varepsilon_1}.
\end{equation}
By \cref{thm:coveringimplieslearning}, the covering scheme guarantees that the backward lightcone $C$ of the output of $\Phi$ is contained entirely within $\Phi$ (in fact, the output state of $\Phi$ is independent of its input state). $C$ is a depth-$(2k+1)d$ circuit. Our final learned state $\hat\rho$ is thus equal to $\Phi(\ketbra{\psi})$, which can be prepared by running $C$ on all-$\ket{0}$ input and trace out those qubits that do not belong to the output of $\Phi$. To achieve $\varepsilon$-closeness in trace distance, it suffices to choose $\varepsilon_1=\frac{\varepsilon^2}{48 n^2}$.

This concludes the proof of the complexity statements in the main result.

\subsection{Optimizing the number of ancilla qubits}
\label{sec:ancilla}
Here we explicitly calculate (a rough upper bound of) the number of ancilla qubits needed for reconstructing $\ket{\psi}$, which corresponds to the red qubits in \cref{fig:1d_rec_circuit}. In~\cref{fig:2dlayer1,fig:2dlayer2,fig:2dlayer3}, the ancilla qubits live in the colored regions outside solid black boxes.

We start from a lattice coloring (e.g.~\cref{fig:2dlatticecoloring}). Note that each small colored region is contained in a $k$-dimensional box of length scale $2kR$; meanwhile each small colored region at least contains a $k$-dimensional box of length scale $R$ (see e.g.~\cite[below Definition 17]{huang2024learning}). The number of colored regions is at most $n/R^k$.

Next we upper bound the number of ancilla qubits that can be associated with a colored region. The length scale of a subset in a covering scheme is at most $2kR+2\times (2k+1)d\leq 3kR$, since we will choose $R$ to be at least $(4k+5)d$. The ancilla qubits live at $2k$ different $k-1$ dimensional surfaces with thickness $(2k+1)d\leq 3kd$; the total volume is at most $(2k)\times (3kR)^{k-1}\times (3kd)$. Overall, the total number of ancilla qubits in the entire circuit is at most
\begin{equation}
    \frac{n}{R^k}\times (2k)\times (3kR)^{k-1}\times (3kd) \leq \frac{n}{R}\times (3k)^{k+1} d.
\end{equation}
Choosing $R=L\cdot (3k)^{k+1} d$ for some sufficiently large constant $L$, the total number of ancilla qubits is an arbitrarily small constant times $n$. In fact we could also afford to choose $L=\omega(1)$ to be some small function (e.g. $\log\log\log n$) so that the total number of ancilla qubits is sublinear. Below we just consider $L$ as being a constant.

The complexity scales with $c$ as shown in~\cref{sec:analysis}, which is given by
\begin{equation}
    c\leq (3kR)^k = O((3k)^{k+2} d)^k. 
\end{equation}

\section{Testing quantum circuit complexity}
The following result is a stronger version than stated in~\cref{cor:testing}.

\begin{thm}
\label{theorem:testing}
    Given copies of an unknown state $\ket{\psi}$ on a $k$-dimensional lattice, with the promise that one of the following two cases hold:
    \begin{itemize}
    \item \textbf{Case 1: Low complexity.} $\ket{\psi}=U\ket{0^n}$ where the depth of $U$ is at most $d$;
    \item \textbf{Case 2: High complexity.} Let $\rho$ be any $n$-qubit state prepared by a depth at most $(2k+1)d$ circuit using $r\cdot n$ ancilla qubits, where $r$ is some small constant. Then $\left\|\rho-\ketbra{\psi}\right\|_1 > \varepsilon$.
\end{itemize}
There is an algorithm that decides which is the case, with success probability at least $1-\delta$, where the sample complexity and running time is the same as in \cref{thm:maindetailed}.
\end{thm}

\begin{proof}
We perform Step 1 and 2 as prescribed by \cref{sec:analysis} (note that $\varepsilon_1=\frac{\varepsilon^2}{48 n^2}$). Then, we declare ``Case 1: Low complexity'' if the search procedure to find all desired approximate local inversions in Step 3 all succeed. Otherwise, declare ``Case 2: High complexity''.

We assume that the reduced density matrices in Step 2 are all within $\varepsilon_1$ error, which happens with probability at least $1-\delta$. Conditioned on this event, the algorithm is always correct, because of the following reasoning.
\begin{itemize}
    \item In Case 1, all search procedures to find the desired approximate local inversions are guaranteed to succeed, so the algorithm must output ``Case 1''.
    \item In Case 2, suppose by contradiction that all search procedures succeed. Then, by the arguments of the main result, this actually gives a \emph{proof} that $\ket{\psi}$ can be prepared within $\varepsilon$ error, using a quantum circuit of depth $(2k+1)d$ and a small amount of ancilla qubits, which contradicts the assumption of Case 2. Therefore, some of the search procedures must fail, and the algorithm must output ``Case 2''.
\end{itemize}
\end{proof}

\section*{Acknowledgements}
This material is based upon work supported by the U.S. Department of Energy, Office of Science, National Quantum Information Science Research Centers, Quantum Systems Accelerator. Additional support is acknowledged from DOE Grant No. DE-SC0024124 and NSF Grant No. 2311733.

\printbibliography

\end{document}